\documentclass[letterpaper, 10 pt, conference]{sigcomm}  

\usepackage{graphics} 
\usepackage{epsfig} 
\usepackage{amsmath} 
\usepackage{amssymb}  
\usepackage{color}
\usepackage{soul}
\usepackage{graphicx}
\usepackage{subcaption}
\usepackage{bm}
\usepackage[ruled, noline, linesnumbered]{algorithm2e}
\newtheorem{thm}{Proposition}

\title{\LARGE \bf
Scalable Routing in SDN-enabled Networks with Consolidated Middleboxes
}

\numberofauthors{3}
\author{
  \alignauthor \ Andrey Gushchin\\
 \affaddr{Cornell University}
  \email{avg36@cornell.edu}
  \alignauthor \ Anwar Walid\\
\affaddr{Bell Labs, Alcatel-Lucent}
  \email{anwar.walid@alcatel-lucent.com}
  \alignauthor \ Ao Tang\\
\affaddr{Cornell University}
  \email{atang@ece.cornell.edu}
}

\begin{document}


\maketitle
\thispagestyle{empty}
\pagestyle{empty}

\begin{abstract}
Middleboxes are special network devices that perform various functions such as enabling security and efficiency.
SDN-based routing approaches in networks with middleboxes need to address resource constraints, such as memory in the switches and processing power of middleboxes, and traversal constraint where a flow must visit the required middleboxes in a specific order. 
In this work we propose a solution based on MultiPoint-To-Point Trees (MPTPT) for routing traffic in SDN-enabled networks with consolidated middleboxes. 
We show both theoretically and via simulations that our solution significantly reduces the number of routing rules in the switches, while guaranteeing optimum throughput and meeting processing requirements. Additionally, the underlying algorithm has low complexity making it suitable in dynamic network environment.
\end{abstract}



\section{INTRODUCTION}

Middleboxes (e.g. proxies, firewalls, IDS, WAN optimizers, etc.) are special network devices that perform functional processing of network traffic in order to achieve a certain level of security and performance.
Each network flow may require certain set of functions. In some cases  these functions can be applied only in a particular order,
which makes routing in networks with middleboxes under limited resources constraints even a more difficult task.
Mechanism of controlling routing through the specified functional sequence is called Service Function Chaining (SFC).
Logically centralized traffic control offered by SDN enables traffic routing optimization (in terms of device costs, total throughput, load balancing, link utilizations, etc.), while satisfying a correct traversal of network middleboxes for each flow.
Several recent works (e.g. \cite{c1}, \cite{c12}, \cite{c13}) provide relevant solutions.


Functionality provided by middleboxes can be incorporated in the network in several ways.
Traditional middlebox is a standalone physical device that can typically perform one network function, and may be located at an ingress switch.
With the development of the Network Function Virtualization (NFV), middleboxes may be implemented using Virtual Machines (VMs) that can be flexibly installed at the Physical Machines (PMs).
In addition, virtualization enables implementation of the consolidated middleboxes \cite{c4}, where a flow receives all of its required service functions at a single machine.
The consolidated middlebox model simplifies traffic routing and helps reduce the number of routing rules in the switches.

In this paper, we follow the model in \cite{c17}, and assume that each middlebox function is an application that can be installed at certain VMs within the PMs.
It is also assumed that every flow obtains all its required functional treatment at a single PM, and thus the consolidated middlebox model is implied in the paper.
Network function consolidation and flexible implementation of middleboxes were previously discussed, for example in \cite{c13}, \cite{c14}, \cite{c15} and \cite{c4}.

Depending on the network traffic environment, two types of routing schemes can be developed: offline, where all required traffic demands are given or can be estimated (for example, using a service level agreement between the customer and the provider), and online, where demands are unknown and a routing solution for each coming flow is made based on the flow class and the current state of the network. 
A solution obtained by a routing scheme can be converted into a set of routing rules that are  installed in the switches.
Different criteria can be used to characterize the achievable performance of a routing scheme: total throughput, average delay, maximum PM utilization, etc.
Besides achieving a desired network performance, a routing scheme must also satisfy resource and routing constraints.
Additionally, three new constraints are of a special interest in the SDN-enabled networks with middleboxes.
\vspace{-3pt}
\begin{itemize}
\item {\bf {Switch memory capacities:}} number of rules installed in a single switch is limited by its memory capacity.
Ternary Content-Addressable Memory (TCAM) used in SDN switches is a scarce resource which is expensive both in terms of cost and power consumption.
\vspace{-7pt}
\item {\bf{Middlebox processing capacities:}} load on each middlebox should not exceed its processing capacity. Overload of middleboxes has to be avoided since it may cause loss of traffic, delay, incorrect traversal sequence or other problems.
\vspace{-7pt}
\item {\bf{Traversal constraints:}} required network functions have to be applied to any given flow in a correct order.   

\end{itemize}

The switch memory constraint is important: flow table overflow is a serious problem that can significantly degrade network performance and, therefore should be avoided. 
Because this constraint is of integer type, it makes the problem of finding an optimal routing solution hard.
If, in addition, middleboxes are added to the network, finding such routing becomes even harder.

In this paper we present an approach based on multi-point-to-point trees that efficiently finds a routing with a guarantee on the maximum number of rules in a single switch, while satisfying all other network constraints.
Moreover, our routing solution scales well with the network size: the explicit bound $C+2|E_{0}|+|V_{T}|-2|V_{pm}|$ on the number of rules is additive and depends linearly on the number of destination nodes ($|V_{T}|$), links ($|E_{0}|$) and flow classes ($C$) in the network.

This paper is organized as follows: in Section 2 we introduce our network model and necessary notations.
In Section 3 we describe our routing solution, in Section 4 we evaluate its performance by simulations and demonstrate its advantages over several other routing schemes. 
Finally, we compare our solution with related works in Section 5, and conclude in Section 6.

\begin{table}
\begin{center}
\begin{tabular}{ |c|c| } 
 \hline
 $m(v_{j})$&capacity of switch $v_{j}\in V_{sw}$ \\  \hline
 $r(v_{j})$ &number of rules in switch $v_{j}\in V_{sw}$\\  \hline
 $b(v_{j})$&capacity of PM $v_{j}\in V_{pm}$ \\  \hline
$g(e)$ &capacity of link $e$ \\  \hline
&  commodity $i$ \\
  $com_{i}$&with source $s_{i}$, destination $t_{i}$,\\
$<s_{i},t_{i},d_{i},c_{i}>$& demand $d_{i}$, and class $c_{i}$ \\ \hline
 $p_{i}$ &cost (in PM resources) of $com_{i}$\\  \hline
$M$ & total number of commodities \\ \hline
$C$ &number of different traffic classes\\  \hline
 $V_{T}$ &set of distinct destinations\\  \hline
\end{tabular}
\end{center}
\caption{Main notations.}
\label{table}
\end{table}


\section{Problem Formulation}

\subsection{Network Topology and Resources}
We assume that the network topology is defined by a directed graph $G_{0}=(V_{0},E_{0})$, where $V_{0}$ is the set of its nodes and $E_{0}$ is the set of edges.
Each node corresponds either to a switch or to a PM, and each edge is a link connecting either two switches, a switch with a PM, or a PM with a switch. 
We denote by $V_{sw}$ and $V_{pm}$ the node sets corresponding to switches and PMs, respectively, so that $V_{sw}\cup V_{pm}=V_{0}$, and $V_{sw}\cap V_{pm}=\emptyset$.
It will be assumed for simplicity that each PM is connected with a single switch by bi-directional links as shown in Fig. \ref{g}. 
Let $V_{sw\rightarrow pm}$ be the subset of nodes in $V_{sw}$ that are directly connected to the PM nodes ($V_{sw\rightarrow pm}=\{sw1, sw2, sw6\}$ in Fig. \ref{g}).

Each switch has a certain memory capacity that can be expressed as a number of rules that it can accommodate. 
We will denote this number by $m(v_{j})$ for a switch located at node $v_{j}$ ($j=1,\dots,|V_{sw}|)$, where $|A|$ is the cardinality of a set $A$.
Additionally, let $r(v_{j})$ be the number of rules in this switch in a routing solution.

Although a PM may have several types of resources (e.g. memory, CPU), it will be assumed for simplicity that each PM is characterized by a single resource capacity that will be denoted by $b(v_{j})$ for a PM located at node $v_{j}$ ($j=1,\dots,|V_{pm}|$).
Similarly, each link $e_{k}\in E_{0}$ $(k=1,\dots, |E_{0}|)$ has an associated link capacity that will be denoted by $g(e_{k})$.

%

\subsection{Network Functions and Commodities}
There exist several types of network functions (firewall, IPS, IDS, WAN optimization, etc.), and each function has its own cost per unit of traffic in terms of PM resources.
Although in this work we assume that this processing cost is the same for all PMs, it is easy to generalize it to the case when the costs are distinct for different PMs.

Additionally, there is a set of $M$ traffic demands or ``commodities'' that have to be routed in the network.
We will use the terms traffic demand and commodity interchangeably.
Commodity $com_{i}$ is defined by a four-tuple $com_{i}=$ $<s_{i},t_{i},d_{i},c_{i}>$, where $i=1,\dots, M$.
Here $s_{i}\in V_{sw}$ and $t_{i}\in V_{sw}$ are, respectively, source and destination nodes, $d_{i}$ is an amount of flow that has to be routed for commodity $com_{i}$, which we will call commodity's demand, and $c_{i}$ is an ordered set of network functions required by this commodity. 
Any such ordered set of network functions defines the class of a commodity.
We will denote by $C$ the total number of different classes of traffic demands.
Due to various functional requirements, different commodities may have different per unit of traffic costs in terms of PM's processing power.
Let $p(i)$ be such cost per unit of traffic for traffic demand $com_{i}$.

Each PM hosts at most $C$ VMs, where a single VM corresponds to a single commodity class. 
It is assumed that when a packet from a commodity of class $k$ arrives to a PM, it is transfered to the virtual machine associated with class $k$, and all network functions of class $k$ are applied to this packet in a correct order.
Distribution of each PM's processing capacity among $C$ VMs has to be determined.
It is assumed, however,  that positions of PMs (nodes $V_{pm}$) are given as an input and are not subject to change. 

By $V_{T}$ we will denote the set of distinct destinations, then $|V_{T}|\leq M$, $|V_{T}|\leq |V_{sw}|$. Main notations are summarized in Table \ref{table}.




%

\subsection{Routing via Integer Linear Optimization}
In this work we employ the idea of consolidated middleboxes, and each packet belonging to $com_{i}$ gets all functional treatment specified by ${c}_{i}$ at a single PM.
It is allowed, however, that a single commodity's traffic is split into several paths from $s_{i}$ to $t_{i}$, and distinct paths may traverse distinct PMs.
We point out that splitting occurs at the IP flow level and not at the packet level.
This is similar to Equal Cost Multipath \cite{c16} in Data Centers, where hashing is used to split traffic at the IP flow level for routing on multiple paths.

If the traffic demands are known in advance, an optimization problem can be posed whose feasible solution defines a routing that satisfies all network constraints.
The variables of this optimization problem $f_{i}^{x}(e)$ are the amount of traffic of commodity $com_{i}$ on edge $e\in E_{0}$, $i=1,\dots, M$. 
Here superscript $x\in\{0,1\}$ with zero value corresponds to the traffic that has not visited a consolidated middlebox, and unit value is used to denote the traffic that has been processed by the required network functions. 
There are thus $2\cdot M\cdot |E_{0}|$ variables in this optimization problem. 
Let $d_{i}(v)$ be the demand from a node $v\in V_{sw}$ for the commodity $com_{i}$.
Note that $d_{i}(v)=d_{i}$ if $v=s_{i}$, and is zero, otherwise.
The problem is formulated as follows.
\setlength{\belowdisplayskip}{0pt} \setlength{\belowdisplayshortskip}{2pt}
\setlength{\abovedisplayskip}{0pt} \setlength{\abovedisplayshortskip}{2pt}
{\center {{\bf{ILP Optimization (1):}}}}
\begin{subequations}
\begin{equation*}
{\bf{min}} \sum\limits_{\substack{e\in E_{0},1\leq i\leq M,\\ x\in \{0,1\}}}f_{i}^{x}(e),
\end{equation*}
\begin{equation*}
\hspace{-3.85cm}\forall v\in V_{sw},\;\forall  i:\;1\leq i\leq M: 
\end{equation*}
\begin{equation}
\hspace{-1.9cm}\sum\limits_{(v,w)\in E_{0}}f_{i}^{0}(v,w)-\sum\limits_{(u,v)\in E_{0}}f_{i}^{0}(u,v)=d_{i}(v),
\label{1a}
\end{equation}
\begin{equation*}
\hspace{-2.55cm}\forall v\in V_{sw}: \;v\neq t_{i}, \; \forall  i: \;1\leq i\leq M: 
\end{equation*}
\begin{equation}
\hspace{-2.5cm}\sum\limits_{(v,w)\in E_{0}}f_{i}^{1}(v,w)-\sum\limits_{(u,v)\in E_{0}}f_{i}^{1}(u,v)=0,
\label{1b}
\end{equation}
\begin{equation}
\hspace{-1.65cm}\forall v\in V_{pm},\;  \forall   i: \;1\leq i\leq M: \hspace{3pt}f_{i}^{1}(u,v)=0, 
\label{1c}
\end{equation}
\begin{equation}
\forall v\in V_{pm},\;  \forall   i: \;1\leq i\leq M: \hspace{3pt}f_{i}^{0}(u,v)=f_{i}^{1}(v,u), 
\label{1d}
\end{equation}
\begin{equation}
\hspace{-1.9cm}\forall v\in V_{pm}: \hspace{3pt}\sum\limits_{1\leq i\leq M}p(i)\cdot f_{i}^{0}(u,v)\leq b(v),
\label{1e}
\end{equation}
\begin{equation}
\hspace{-1.05cm}\forall e\in E_{0}:\hspace{3pt}\sum\limits_{1\leq i\leq M}f_{i}^{0}(e)+\sum\limits_{1\leq i\leq M}f_{i}^{1}(e)\leq g(e),
\label{1f}
\end{equation}
\begin{equation}
\hspace{-4.2cm}\forall v\in V_{sw}:\hspace{3pt} r(v)\leq m(v),
\label{1g}
\end{equation}
\begin{equation}
\hspace{-0.22cm}\forall e\in E_{0}, \; \forall i:\; 1\leq i\leq M, \; \forall x\in \{0,1\}:\hspace{3pt}f_{i}^{x}(e)\geq 0.
\label{1h}
\end{equation}
\end{subequations}

\begin{figure}[t]
\begin{center}
  \begin{subfigure}[t]{0.2\textwidth}
      \includegraphics[width=\textwidth]{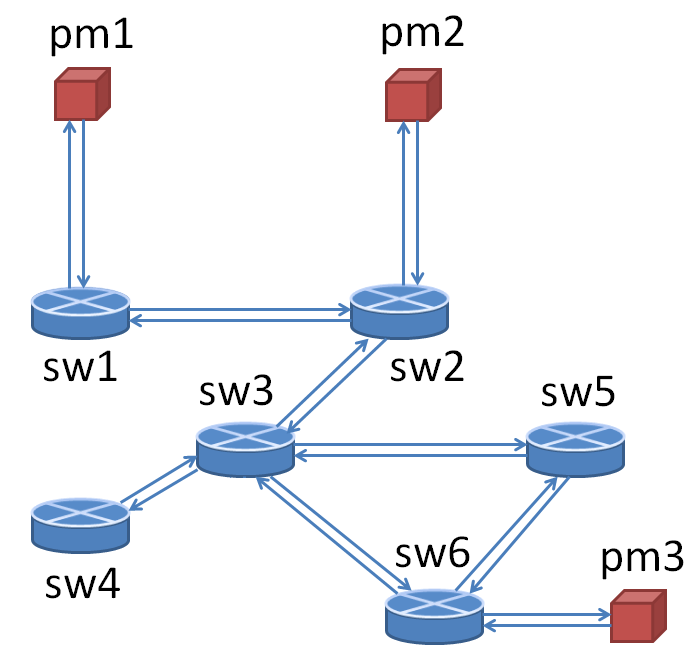}
      \caption{Example of a real network topology defined by graph $G_{0}=(V_{0},E_{0})$.}  
      \label{g}    
\end{subfigure}%
~~~~~
\begin{subfigure}[t]{0.2\textwidth}
      \includegraphics[width=\textwidth]{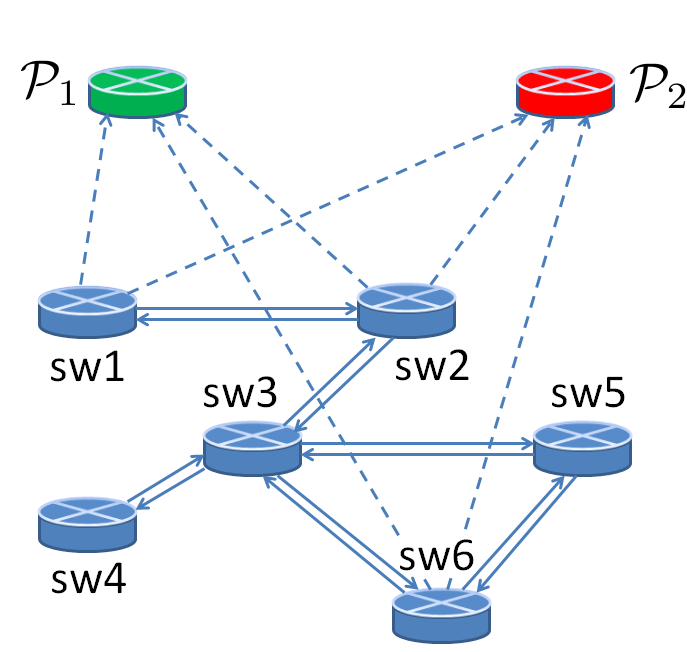}
      \caption{Grpah $G_{1}=(V_{1},E_{1})$ constructed from graph $G_{0}$ at Step 1.}
      \label{g1}          
\end{subfigure}
~~~~~
\begin{subfigure}[t]{0.18\textwidth}
      \includegraphics[width=\textwidth]{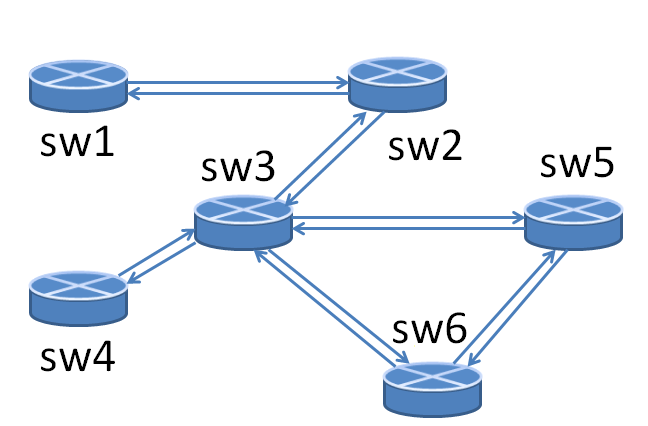}
      \caption{Graph $G_{2}=(V_{2},E_{2})$ constructed from graph $G_{1}$ at Step 2.}
      \label{g2}          
\end{subfigure}
\end{center}
\caption{Example of a given graph $G_{0}$ and graphs $G_{1}$ and $G_{2}$ constructed at the first and the second steps of our algorithm, respectively.}
\label{ba}

   \end{figure}

Constraints \eqref{1a} and \eqref{1b} are flow conservation constraints for switches, constraint \eqref{1c} forbids the traffic that has already been processed by a middlebox (PM), to visit a middlebox again. 
Next constraint \eqref{1d} says that all unprocessed traffic becomes processed at the PM associated with node $v\in V_{pm}$. 
Further, constraint \eqref{1e} is a PM processing capacity constraint. 
The following constraint \eqref{1f} is a link capacity constraint, and condition \eqref{1g} corresponds to the switch memory constraint.
Finally, \eqref{1h} requires that all flow values are nonnegative. 
The objective function of this optimization problem is the total flow over all edges. 
This choice of the objective function guarantees that no cycles will exist in an optimal solution.
Notice that there is no constraint $f_{i}^{0}(v,u)=0$ similar to constraint \eqref{1c}, because it will be automatically satisfied due to the optimization's objective function.

Solution to this optimization problem expressed in terms of variables $f_{i}^{x}(e)$ can be translated to a path-flow formulation \cite{c10}, and the routing rules in switches can be obtained that implement this path-flow solution. 
Each routing rule in a switch corresponds to a single path in the path-flow solution.
Notice that in the solution to the optimization problem, more than one source-destination path can be used to transfer traffic for a single commodity.

The optimization problem formulated above contains integer switch memory constraints \eqref{1g} and thus belongs to the class of Integer Linear Programs (ILP).
This problem, therefore, is NP-hard, and it is extremely difficult to obtain its solution.
In this work, we adapt the idea of  Multipoint-To-Point Trees to construct a feasible routing scheme for SDN-enabled networks with middleboxes and known traffic demands. 
Although the integer switch memory constraints are not explicitly incorporated into our solution, we can obtain the worst case bound on the number of rules in each switch. 
Moreover, we show that this bound scales well with the network size and is low enough for our routing scheme to be implemented in the networks with existing switches.


\section{Solution Overview}

\subsection{MPTPT Approach}
In this work we take advantage of the capabilities provided by SDN to design efficient routing. 
In particular, SDN facilitates global design optimization based on inputs and measurements collected from various points of the network, and the ability to translate design solutions into rules which can be downloaded to the switches.
One of the major components of our routing solution is multipoint-to-point trees that were previously used, for example, by the label based forwarding mechanism of MPLS \cite{c11}.
Each multipoint-to-point tree is rooted at some node, and all its edges are oriented towards this root node.
Such trees can be used to route traffic from several sources to a single destination, and each tree is assigned with its own tag which is used to label all traffic belonging to this tree.
Utilization of MPTPTs helps to reduce the number of routing rules in the whole network \cite{c6}.

Our solution contains two main steps. 
These steps are purely computational (not actual routing steps), and allow to determine how the traffic for each commodity is labeled and routed.
At the first step we route all traffic from the sources $s_{i}$, ($i=1,\dots,$ $M$) to PMs.
At the second step, we route all traffic that has been processed by the required network functions during the first step from the PMs to the corresponding destinations $t_{i}$, ($i=1,\dots, M$).
Both steps involve construction of MPTP trees: there are $C$ roots for multipoint-to-point trees built at the first step, where each root corresponds to a particular flow class, and there are
$|V_{T}|$ roots for the trees at the second step. 
There can be in general more than one MPTP tree rooted at a single node.
In Fig. \ref{diagram} we show the schematic of our MPTPT-based routing algorithm.

\subsection{Step 1: Routing from Sources to PMs}

At the first step we consider a graph $G_{1}=(V_{1},E_{1})$ which is obtained from the initial graph $G_{0}$ as follows: we add $C$ additional nodes $\mathcal{P}_{1},\dots, \mathcal{P}_{C}$ such that node $\mathcal{P}_{k}$ corresponds to the traffic class $k$.
This set of $C$ new nodes is denoted by $V_{\mathcal{P}}$, and $|V_{\mathcal{P}}|=C$.
We further remove "PM" nodes belonging to the set $V_{pm}$, together with the edges going to and from these nodes.
Then, we connect each node from $V_{sw\rightarrow pm}$ by edges to every node from $V_{\mathcal{P}}$.
These new edges are not assigned with capacities explicitly, but the maximum amount of flow on them will be determined by the capacities of PMs and the capacities of removed links from graph $G_{0}$ that were connecting nodes in $V_{sw\rightarrow pm}$ with nodes $V_{pm}$.
The vertex set of graph $G_{1}$ is a union of node sets $V_{sw}$ and $V_{\mathcal{P}}$: $V_{1}=V_{sw}\cup V_{\mathcal{P}}$.
Number of links in the graph $G_{1}$ is $|E_{1}|=|E_{0}|+|V_{pm}|\cdot (C-2)$.
In Fig. \ref{ba} we show an example of a network topology defined by a graph $G_{0}$ (Fig. \ref{g}) and corresponding constructed graph $G_{1}$ (Fig. \ref{g1}).
In this example it is assumed that there are two classes of flows and the nodes $\mathcal{P}_{1}$ and $\mathcal{P}_{2}$ are associated with flow classes one and two, respectively.
In Fig. \ref{g1} the new added links are shown by dashed arrows.

We additionally modify destinations of the given commodities. 
In particular, destination of all traffic demands of class $k$ is node $\mathcal{P}_{k}$, $k=1,\dots, C$. 
Therefore, for each commodity $com_{i}$, its destination is one of the nodes in $V_{\mathcal{P}}$. 
We can now formulate an LP optimization problem that we solve at the first step of our method.
In contrast to the commodity-based ILP problem considered in the previous subsection, the optimization here is in a tree-based formulation, and we do not distinguish traffic from different sources if they are for the same destination, i.e. if they belong to the same network class.
Let $\hat{v}$ denote a PM connected to node $v\in V_{sw\rightarrow pm}$ in graph $G_{0}$ (for example, $\hat{v}=pm3$ for $v=sw6$ in the example from Fig. \ref{ba}), and $p(t)$, where $t\in V_{\mathcal{P}}$, denotes the cost of PM resources per unit of traffic of class corresponding to the node $t$.

 \begin{figure}[t!]
      \centering
      \includegraphics[scale=0.4]{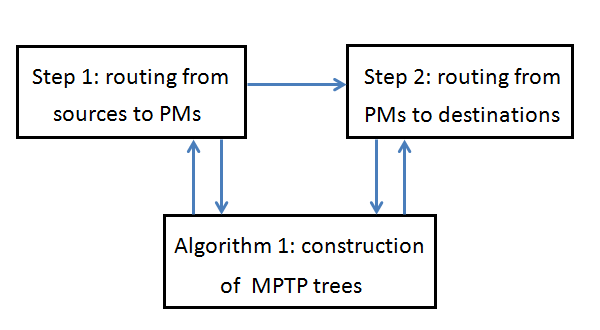}
      \caption{Schematic of the MPTPT-based routing algorithm.}
      \label{diagram}
   \end{figure}

\setlength{\belowdisplayskip}{0pt} \setlength{\belowdisplayshortskip}{0pt}
\setlength{\abovedisplayskip}{0pt} \setlength{\abovedisplayshortskip}{0pt}
{\center {{\bf{LP Optimization (2) of Step 1:}}}}
\begin{subequations}
\begin{equation*}
{\bf{min}} \sum\limits_{e\in E_{1}, t\in V_{\mathcal{P}}}f_{t}(e),
\end{equation*}
\begin{equation*}
\hspace{-4.15cm}\forall  t\in V_{\mathcal{P}}, \; \forall v\in V_{1}, \; v\neq t: 
\end{equation*}
\begin{equation}
\hspace{-2cm}\sum\limits_{(v,w)\in E_{1}}f_{t}(v,w)-\sum\limits_{(u,v)\in E_{1}}f_{t}(u,v)=d_{t}(v),
\label{2a}
\end{equation}
\begin{equation}
\hspace{-3.03cm}\forall e\in E_{1}\cap E_{0}:\hspace{3pt}\sum\limits_{t\in V_{\mathcal{P}}}f_{t}(e)\leq g(e),
\label{2b}
\end{equation}
\begin{equation}
\hspace{-0.245cm}\forall v\in V_{sw\rightarrow pm}:\hspace{3pt}\sum\limits_{t\in V_{\mathcal{P}}}f_{t}(v,t)\leq \min\bigl\{g(v,\hat{v}),g(\hat{v},v)\bigl\},
\label{2c}
\end{equation}
\begin{equation}
\hspace{-1.85cm}\forall  v\in V_{sw\rightarrow pm}:\hspace{3pt}\sum\limits_{t\in V_{\mathcal{P}}}p(t)\cdot f_{t}(v,t)\leq b(\hat{v}),
\label{2d}
\end{equation}
\begin{equation}
\hspace{-3.5cm}\forall e\in E_{1}, \; \forall t\in V_{\mathcal{P}}:\hspace{3pt}f_{t}(e)\geq 0.
\label{2e}
\end{equation}
\end{subequations}
\vspace{2pt}

In this optimization problem variable $f_{t}(e)$ is an amo-unt of flow to destination $t\in V_{\mathcal{P}}$ on link $e\in E_{1}$.
Constraint \eqref{2a} is a flow conservation at node $v$, condition \eqref{2b} is a link capacity constraint that should be satisfied for any link that belongs to the both edge sets $E_{0}$ and $E_{1}$ of graphs $G_{0}$ and $G_{1}$, respectively.
Further, constraint \eqref{2c} is a link capacity constraint for the links that connect switches with PMs in graph $G_{0}$. 
This constraint is necessary for feasibility of the solution to optimization problem (2) in the original graph $G_{0}$.
Notice that in the right hand side of \eqref{2c} there is a minimum between capacities of the links going from a switch to a PM and from a PM to a switch. 
It will guarantee that all traffic processed at a PM can be send back to a switch connected to this PM.
Next constraint \eqref{2d} is a PM capacity constraint, and by \eqref{2e} we require that flow on each link is nonnegative. 
As in the ILP optimization problem (1), we minimize the total network flow to avoid cycles.

Solution to the optimization problem (2) determines how the traffic is routed from the sources to the PMs. 
Using Algorithm Flow2Trees$(t)$ from \cite{c6} that is listed as Algorithm \ref{A1} below for completeness, from a basic feasible solution \cite{c10} $f_{t}(e)$ to the LP (2) we construct multipoint-to-point trees rooted at the destination  nodes from $V_{\mathcal{P}}$, so that all network traffic in the solution is distributed among these trees. 
Each tree contains traffic of the same class, leafs of a tree are the sources for this traffic class, and amount of traffic from each source in any tree can be determined. 
It is possible that several $V_{sw\rightarrow pm}$ nodes belong to the same tree, i.e. one tree can route traffic to several PMs.
Algorithm \ref{A1} is iteratively applied to construct trees to each destination $t\in V_{\mathcal{P}}$.
We will provide an upper bound on a total number of trees in the subsection $3.5$.
We refer the reader to \cite{c6} for the details and analysis of Algorithm \ref{A1}.
\begin{algorithm}
\PrintSemicolon
\SetKwInOut{Input}{Input}\SetKwInOut{Output}{Output}
\Input{$G=(V ,E)$, $t$, $f_{t}(e)$ ($\forall e\in E$).}
\Output{Set of MPTP trees rooted at $t$ and containing all traffic to $t$.}
\While {\mbox{there is a source $s$ with demand to $t$}}{
using only edges $e$ with flow to $t$ ($f_{t}(e)>0)$, construct a tree $R$ to $t$ spanning all sources with demand to $t$\;
move as much flow as possible to $R$\;
}
\caption{Flow2Trees$(t)$\label{A1}}
\end{algorithm}

\subsection{Step 2: Routing from PMs to Destinations}

At the second step of our algorithm we use MPTP trees to route traffic from the PMs to destinations in graph $G_{2}$ obtained from $G_{1}$ as follows.
First, nodes $V_{\mathcal{P}}$ and links to them are removed from the network.
Therefore, the node set of the resulting graph $G_{2}=(V_{2},E_{2})$ only contains nodes from $V_{sw}$: $V_{2}=V_{sw}$.
Number of links in graph $G_{2}$ is $|E_{2}|=|E_{1}\cap E_{0}|=|E_{0}|-2\cdot |V_{pm}|$.
Second, the link capacities are updated: for each link $e$, the amount of traffic on it  in the solution to (2) is subtracted from this link's initial capacity $g(e)$.
We will denote by $\bar{g}(e)$ the updated capacity of link $e$.
Graph $G_{2}$ corresponding to graph $G_{0}$ from Fig. \ref{g} is shown in Fig. \ref{g2}.

\begin{figure}[t]
\begin{center}
  \begin{subfigure}[t]{0.3\columnwidth}
      \includegraphics[width=\columnwidth]{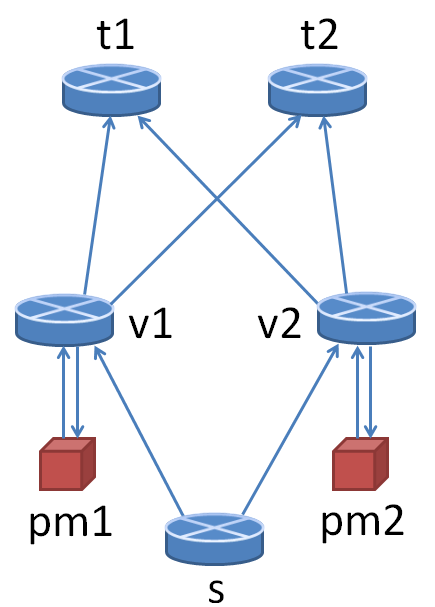}
      \caption{Network topology, $G_{0}=(V_{0},E_{0})$.}  
      \label{ambig1}    
\end{subfigure}%
~~~~~
\begin{subfigure}[t]{0.3\columnwidth}
      \includegraphics[width=\columnwidth]{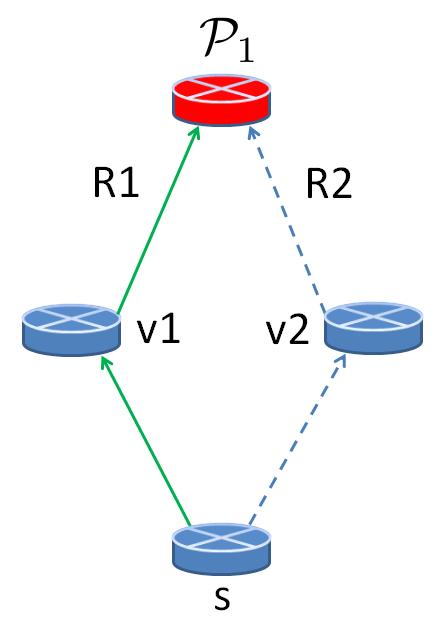}
      \caption{Graph $G_{1}=(V_{1},E_{1})$ and trees $R_{1}$, $R_{2}$ obtained at Step 1.}
      \label{ambig2}  
\end{subfigure}
\end{center}
\caption{Example that shows possible ambiguity in commodity assignment at Step 2.}
\label{ambig}

   \end{figure}

We then create a set of commodities for the second step.
It is assumed that all traffic processed at a PM $\hat{v}$ returns to switch $v\in V_{sw\rightarrow pm}$ connected to it.
Therefore, all traffic at Step 2 is routed from the nodes $V_{sw\rightarrow pm}$ to the destinations $t_{i}$, where $i=1,\dots, |V_{T}|$.
Solution to optimization (2) determines amount of traffic of every class and from every source arriving to each PM.
However, amount of traffic to each destination $t_{i}$ arriving to a PM, in some cases can not be determined unambiguously.
This can happen when there exist more than one commodities with the same source and of the same class but with different destinations.
We illustrate this possibility with an example from Fig.~\ref {ambig}. 
In Fig.~\ref{ambig1} the network topology is shown: there is only one source node $s$, two PMs and two destination nodes $t1$ and $t2$. 
It is also assumed that there is only one traffic class 1, $30$ units of traffic from $s$ should be sent to $t1$, and $70$ units to $t2$.
Graph $G_{1}$ constructed at the first step of our algorithm is shown in Fig. \ref{ambig2}.
Suppose that two trees to node $\mathcal{P}_{1}$ were obtained at the first step: tree $R1$ ($s\rightarrow v1\rightarrow \mathcal{P}_{1})$, and tree $R2$ ($s\rightarrow v2\rightarrow \mathcal{P}_{1})$.
The links belonging to the trees $R1$ and $R2$ are shown by solid green (tree $R1$) and dashed blue (tree $R2$) lines in Fig. \ref{ambig2}.
Assume for example, that $40$ units of traffic of class 1 belong to the tree $R1$, and $60$ units belong to the tree $R2$.
Therefore, after Step 1 of our algorithm it is known how much traffic of this class from source $s$ arrives to node $v1$ (to be processed at PM1), and how much traffic arrives to node $v2$ (to obtain functional treatment at PM2), but distribution of traffic by destination at nodes $v1$ and $v2$ is unknown.
This information, however, is necessary to define commodities at the second step of our approach, and thus a distribution decision is required.

We will use the following heuristic to determine the traffic distribution by destination at each node $v\in V_{sw\rightarrow pm}$.
Let $R$ be a set of trees obtained at Step 1 of our algorithm that carry traffic of the same class $c$ from a source node $s$  to the root node $\mathcal{P}_{c}$ corresponding to this traffic class. 
In addition, let $T$ be the set of destinations of commodities with source $s$ and of class $c$, and $d_{1},\dots, d_{|T|}$ are corresponding demands.
By the definition of a tree, in each tree $R_{i}$ from set $R$, there is a unique path from $s$ to $\mathcal{P}_{c}$, and therefore, all traffic from $s$ in the same tree obtains functional treatment at a single PM.
According to our heuristic, in each tree $R_{i}$, amount of traffic to destination $t\in T$ is proportional to the fraction of traffic to this destination in the total amount of traffic to all destinations, i.e. proportional to $d_{t}/\sum\limits_{i=1}^{|T|}d_{i}$.
In example from Fig. \ref{ambig}, $R=\{R_{1},R_{2}\}$, $\mathcal{P}_{c}=\mathcal{P}_{1}$, $T=\{t1,t2\}$, $d_{1}=30$ and $d_{2}=70$.
Then, according to the heuristic, in tree $R_{1}$: $30/100\cdot 40=12$ units of traffic are to destination $t_{1}$, $70/100\cdot 40=28$ units of traffic are to destination $t_{2}$.
Similarly, in tree $R_{2}$ the distribution is $30/100\cdot 60 = 18$ and  $70/100\cdot 60 = 42$ units to $t_{1}$ and $t_{2}$, respectively.

Using this distribution heuristic, we form a set of commodities for the second step of our algorithm.
At the Step 2 we do not distinguish traffic from different sources and from different network classes if they have the same destination.
We construct MPTP trees with the roots at the destinations $t_{i}$, $i=1,\dots, |V_{T}|$.
Similarly to Step 1, we first solve the following LP:

{\center {{\bf{LP Optimization (3) of Step 2:}}}}
\begin{subequations}
\begin{equation*}
{\bf{min}} \sum\limits_{e\in E_{2}, t\in V_T}f_{t}(e),
\end{equation*}
\begin{equation*}
\hspace{-4cm}\forall t\in V_{T}, \; \forall v\in V_{2}, \; v\neq t: 
\end{equation*}
\begin{equation}
\hspace{-2cm}\sum\limits_{(v,w)\in E_{2}}f_{t}(v,w)-\sum\limits_{(u,v)\in E_{2}}f_{t}(u,v)=d_{t}(v),
\label{3a}
\end{equation}
\begin{equation}
\hspace{-3.61cm}\forall e\in E_{2}:\hspace{3pt}\sum\limits_{t\in V_T}f_{t}(e)\leq \bar{g}(e),
\label{3b}
\end{equation}
\begin{equation}
\hspace{-3.3cm}\forall e\in E_{2}, \; \forall t\in V_T:\hspace{3pt}f_{t}(e)\geq 0.
\label{3c}
\end{equation}
\end{subequations}

\vspace{3pt}
Here \eqref{3a} and \eqref{3b} are flow conservation and link capacity constraints, respectively, and \eqref{3c} is a requirement for flows to be non negative on each link. 
Using a basic feasible solution to this problem, we apply again Algorithm \ref{A1} and obtain another set of multipoint-to-point trees.
Complete version of our MPTPT-based routing approach is summarized in Algorithm \ref{A2}.

\begin{algorithm}
\DontPrintSemicolon
\SetKwInOut{Input}{Input}\SetKwInOut{Output}{Output}
\Input{$G_{0}=(V_{0},E_{0})$, commodities $com_{i}$ ($i=1,\dots, M$). }
\Output{Set of MPTP trees rooted at PM nodes and destination nodes.}
{\bf{Step 1: routing from sources to PMs:}}\\
\Indp  construct graph $G_{1}=(V_{1},E_{1})$ from $G_{0}=(V_{0},E_{0})$;\\
	obtain commodities for Step 1;\\
find a basic feasible solution to LP (2);\\
find MPTP trees for the solution to LP (2) using Algorithm \ref{A1};\\
\Indm
{\bf{Step 2: routing from PMs to destinations:}}\\
\Indp 
construct graph $G_{2}=(V_{2},E_{2})$ from $G_{1}=(V_{1},E_{1})$;\\
obtain commodities for Step 2;\\
find a basic feasible solution to LP (3);\\
find MPTP trees for the solution to LP (3) using Algorithm \ref{A1}.\\

\caption{MPTPT-Based Routing\label{A2}}
\end{algorithm}

After both steps of our algorithm are performed, we can determine for any initial commodity $<s_{i}, t_{i},d_{i}, c_{i}>$ what trees carry its traffic to the destination $t_{i}$.
Each commodity's packet is assigned with two tags at the source switch: one for a tree label from Step 1, and another one for a tree label from Step 2.
The first label can be removed from a packet during functional processing at a PM, and therefore the maximum number of routing rules in a single switch does not exceed the total number of multipoint-to-point trees of both steps.
As suggested in previous works (e.g. \cite{c1}), VLAN and ToS fields of a packet header can be used for labels.


\subsection{Analysis}
In this subsection we provide and prove an upper bound on the total number of MPTP trees generated by Algorithm \ref{A2}.
Each tree has its own label and any switch may contain at most one routing rule corresponding to this tree.
The bound, therefore, also limits the number of routing rules in any switch.
\begin{thm}
Number of MPTP trees produced by Algorithm \ref{A2} does not exceed $C+2|E_{0}|+|V_{T}|-2|V_{pm}|$.
\end{thm}
\begin{proof}
It was shown in \cite{c6} that when Algorithm \ref{A1} is iteratively applied to a basic feasible solution of the multicommodity flow problem (3), the maximum possible number of created trees is $|V_{T}|+|E_{2}|$, i.e. bounded above by the sum of number of destinations and number of links in a network. 
The second term in this sum ($|E_{2}|$) corresponds to the number of bundle constraints in LP. 
A constraint is called bundle if it involves variables for different destinations.
In optimization problem (3) link capacity constraints \eqref{3b} are bundle, and there are $|E_{2}|$ such constraints.
Although optimization problem (2) is slightly different from (3), a similar bound for it can also be established.
Number of bundle constraints in (2) is $|E_{0}|-2\cdot |V_{pm}|+|V_{pm}|+|V_{pm}|=|E_{0}|$, and number of destinations is equal to the number of traffic classes $C$.
Therefore, the total number of trees produced by Algorithm \ref{A2} is $C+|E_{0}|+|V_{T}|+|E_{2}|=C+2|E_{0}|+|V_{T}|-2|V_{pm}|$.
\end{proof}
Notice that while our bound depends on the number of classes $C$, it does not depend on the number of commodities, because $|V_{T}|$ is bounded by $|V_{sw}|$.
The bound is additive and thus scales well with the network size. 
Moreover, as shown by simulations, the real number of routing rules obtained by our algorithm is generally much smaller than this worst case bound.
It is crucial that a basic feasible solution is used as an input to the Algorithm \ref{A1} at both steps of Algorithm \ref{A2}.
We refer the reader to  \cite{c10} and \cite{c6} for a more detailed discussion of basic feasible solutions and bundle constraints.

\begin{figure*}[t]
\begin{center}
  \begin{subfigure}[t]{0.25\textwidth}
      \includegraphics[width=\textwidth]{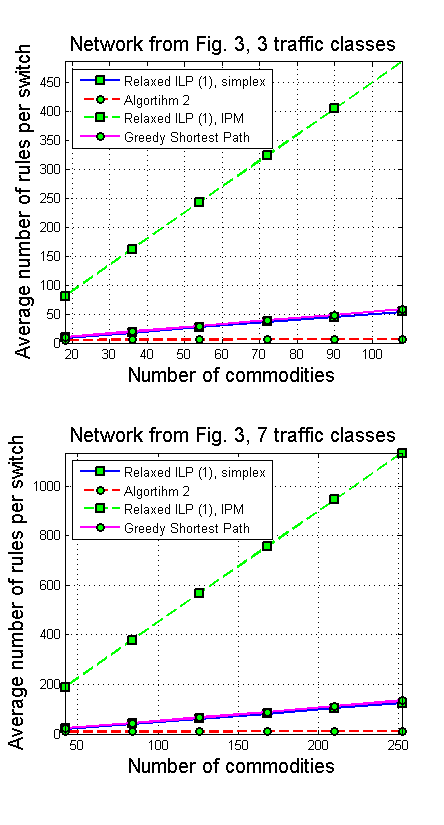}
      \caption{Network topology from Fig. \ref{ba}.}  
      \label{exp1_a}    
\end{subfigure}%
~~~~~
\begin{subfigure}[t]{0.25\textwidth}
      \includegraphics[width=\textwidth]{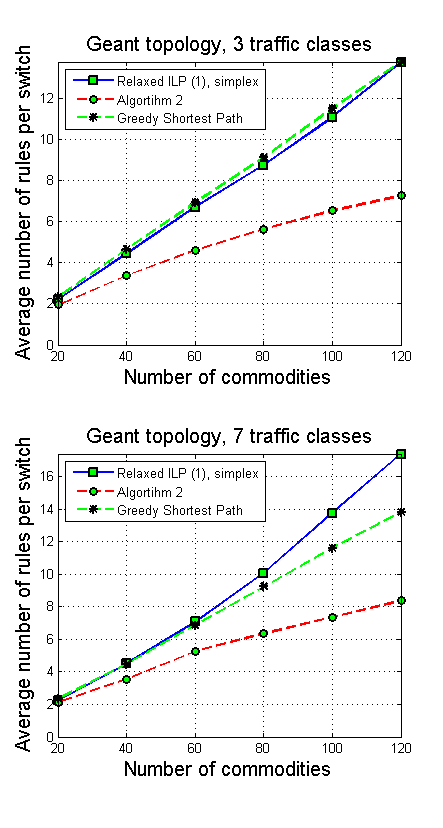}
      \caption{Geant topology.}
      \label{exp1_b}          
\end{subfigure}%
~~~~~
\begin{subfigure}[t]{0.255\textwidth}
      \includegraphics[width=\textwidth]{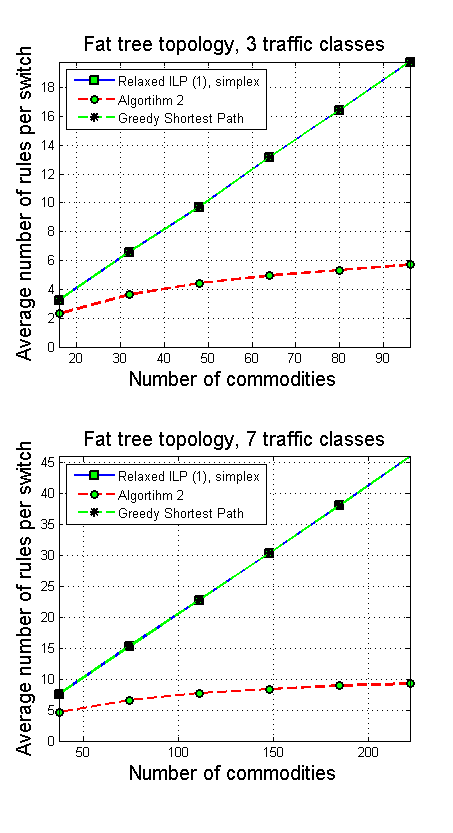}
      \caption{Fat tree topology.}
      \label{exp1_c}          
\end{subfigure}
\end{center}
\caption{Comparison of average numbers of routing rules in switches for three and seven traffic classes.}
\label{exp1}
   \end{figure*}

Therefore, Algorithm 2 efficiently solves a routing problem (it contains two linear optimizations and Algorithm 1 with polynomial time complexity) with a guarantee that the number of routing rules in each switch is limited by an additive bound.

%


\section{Evaluation}

In this section we evaluate the performance of Algorithm \ref{A2} and compare it with three other routing schemes. 
The first routing scheme is defined by optimization problem (1) with relaxed integer switch memory constraint, and a basic feasible solution for it is found using simplex method.
The second scheme uses the same relaxed LP, but an interior point method (IPM) is applied to find a solution.
Finally, the third scheme is based on a greedy shortest path approach.
In this approach the commodities are initially sorted in descending order by their total PM capacity requirement.
Then, iteratively for each commodity a shortest path is found from its source to a PM, and then a shortest path from the PM to commodity's destination.
If link and PM capacity constraints on the shortest path do not allow to send commodity's total demand, a maximum possible fraction of it is sent along this path, and the remaining traffic is sent along the next shortest paths until all commodity's demand is routed. 
If at some point there is no path available to send commodity's residual demand, the algorithm stops.

Our evaluation analysis consists of two experiments. 
In the first experiment we find routing solution using each of four algorithms and calculate an average number of routing rules in switches for each solution.
Second experiment allows to estimate for each routing algorithm the maximum total throughput that it can route.
Both experiments are carried out for three network topologies: example from Fig. \ref{ba}, Geant topology, and fat tree topology.
Geant network contains 41 switch nodes and 9 additional PM nodes that are connected to 9 switch nodes having the highest nodal degree (so that each PM is connected to exactly one switch).
Fat tree topology consists of 22 switch nodes (2 core, 4 aggregation and 16 edge switches), and 6 PM nodes such that each PM node is connected to a single core or aggregate switch node.
Link and PM capacities were fixed in each simulation, and took values, respectively, 100 and 500 for the network on Fig. \ref{ba}, and 500, 500 for Geant topology.
For the fat tree topology links between core and aggregation switches had capacities 200, links between aggregation and edge switches had capacities 10, and links between switches and PMs were fixed at 100.
In addition, each PM had capacity 500.

{\bf{Experiment 1: Average Number Of Routing Rules.}}
In the first experiment we varied number of classes and number of commodities, and each commodity's source, destination and class were generated randomly. 
The demands of the commodities, however, were all equal and fixed at 0.2.
Results of Experiment 1 for 3 and 7 traffic classes are shown in Fig. \ref{exp1}.
It can be observed from the results that Algorithm \ref{A2} allows to reduce average number of routing rules in switches by a factor of up to 10.
We did not add plots corresponding to the interior point method solution for Geant and fat tree topologies because in the IPM solution average number of rules is much higher compared to the other algorithms.
We also performed simulations for one and five traffic classes, and the results look similar to Fig. \ref{exp1}.
The values of bounds on the maximum number of rules in switches are 43, 295 and 137 for the topologies in the same order they are presented in Fig. \ref{exp1} and for 7 traffic classes.
These values were obtained under assumption that $|V_{T}|=|V_{sw}|$ and therefore, limit the number of routing rules in each switch for any arbitrary large number of commodities.

%

\begin{figure*}[t]
\begin{center}
  \begin{subfigure}[t]{0.26\textwidth}
      \includegraphics[width=\textwidth]{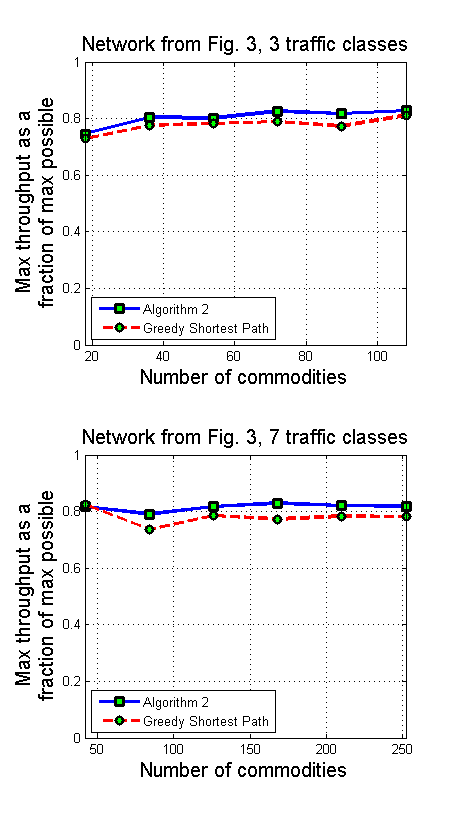}
      \caption{Network topology from Fig. \ref{ba}.}  
      \label{exp2_a}    
\end{subfigure}%
~~~~~
\begin{subfigure}[t]{0.26\textwidth}
      \includegraphics[width=\textwidth]{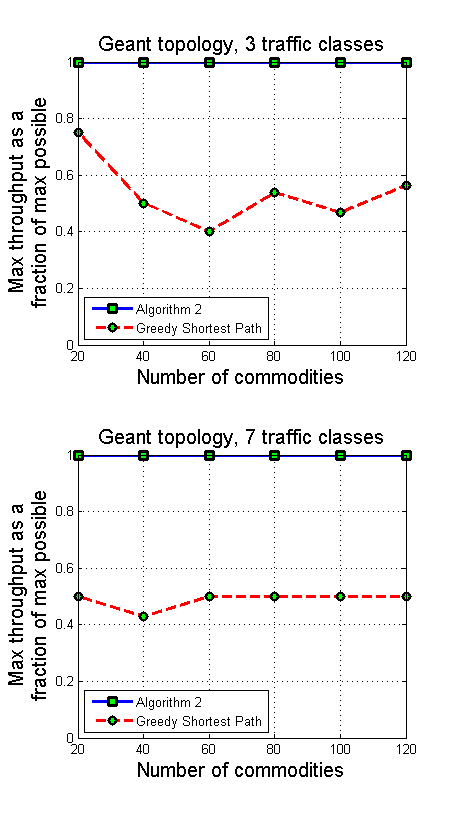}
      \caption{Geant topology.}
      \label{exp2_b}          
\end{subfigure}%
~~~~~
\begin{subfigure}[t]{0.26\textwidth}
      \includegraphics[width=\textwidth]{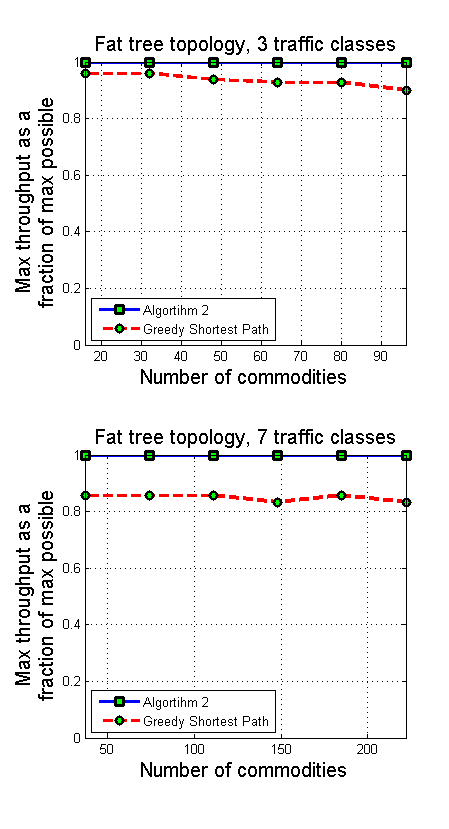}
      \caption{Fat tree topology.}
      \label{exp2_c}          
\end{subfigure}
\end{center}
\caption{Comparison of maximum total throughputs for three and seven traffic classes.}
\label{exp2}
   \end{figure*}

{\bf{Experiment 2: Maximum Total Throughput.}}
In the second experiment we measured the maximum total throughput that can be routed in a network by the Algorithm \ref{A2}.
Notice, that ILP (1) with relaxed switch memory constraint always finds a routing solution when it exists.
Therefore, we used the relaxed LP (1) to determine the maximum possible network throughput.
For a given set of commodities, we increased iteratively demands of all commodities by the same value until the relaxed LP (1) became unfeasible.
We stored this maximum demand value, and repeated the procedure for the Algorithm \ref{A2} and also for the Greedy Shortest Path algorithm.
Results provided in Fig. \ref{exp2} demonstrate that a loss in maximum throughput of the Algorithm \ref{A2} is relatively small.

\section{Related Work}

Most of the previous works on routing in networks with middleboxes aim to achieve a fair load balance among middleboxes \cite{c1}, \cite{c4}. 
However, it is generally assumed that for each given commodity a set of proper paths is provided (or a single path) \cite{c1}, \cite{c4}, \cite{c2}. 
Although this assumption simplifies finding a routing, it has an important disadvantage: it is generally not easy to find a set of suitable paths for all commodities such that all traffic demands can be routed and network constraints are satisfied.
In \cite{c17} the problem setup is similar to ours: middleboxes run as virtual machines at the PMs, traffic demands are known, and a routing linear optimization program is proposed.
The integer switch memory constraint, however, is not incorporated into the routing problem, what makes it easier to find a feasible solution.
The authors also explore the problem of an optimal placement of middleboxes.

An optimization model for Network Intrusion Detection Systems (NIDS) load balancing is presented in \cite{c2}. 
A linear optimization problem contains estimates of the commodities' demands and thus is designed for a carrier-grade traffic environment. 
It is solved periodically to remain an optimal feasible routing for changing traffic demands. 
This optimization problem does not include switch capacity constraints, and its goal is to balance the load among different NIDS.
It is also assumed that for each commodity a precomputed path is given as an input. 
Therefore, the problem solved in \cite{c2} is not exactly a traffic engineering problem, but a load engineering.

Switch memory constraint was taken into account in \cite{c1}, and it was also assumed that there exists a set of suitable paths for each commodity.
Due to the difficulty of the original optimization problem,  it was decomposed in \cite{c1} into two stages. 
At the first (offline) stage, only the switch memory constraint is taken into account, and for each commodity a subset of paths is chosen from its original path set. 
Since the switch memory constraint is integer, the whole offline optimization problem solved at the first stage is an Integer Linear Program (ILP). 
At the second (online) stage, a simpler Linear Program (LP) is formulated to solve a load balancing problem.

\section{CONCLUSION}
In this work we proposed a multipoint-to-point tree based algorithm for SDN-enabled networks with middleboxes and given required traffic demands.
We showed both theoretically and experimentally that in the routing solution obtained by our algorithm, the maximum number of routing rules in a single switch is bounded, and this explicit bound scales well with the network size.
Moreover, the low complexity of the algorithm allows its application the algorithm in dynamic network environment.


\small
\begin{thebibliography}{99}
\vspace{5pt}
\bibitem {c10} Ahuja, R.K.,  Magnanti T.L., and Orlin, J.B. "Network flows: theory, algorithms, and applications." (1993).

\bibitem{c15} Anderson, J. W., Braud, R., Kapoor, R., Porter, G., and Vahdat, A. (2012, October). xOMB: extensible open middleboxes with commodity servers. In ACM/IEEE symposium on Architectures for networking and communications systems (pp. 49-60). ACM.

\bibitem{c6} Applegate, D., and Thorup, M. (2003). Load optimal MPLS routing with N+ M labels. In INFOCOM 2003. IEEE Societies (Vol. 1, pp. 555-565). IEEE.

\bibitem{c17} Charikar, M., Naamad, Y., Rexford, J., and Zou, X. K. Multi-Commodity Flow with In-Network Processing.

\bibitem{c16} Chiesa, M., Kindler, G., and Schapira, M. Traffic Engineering with ECMP: An Algorithmic Perspective.



\bibitem{c12} Fayazbakhsh, S. K., Sekar, V., Yu, M., and Mogul, J. C. (2013, August). FlowTags: enforcing network-wide policies in the presence of dynamic middlebox actions. In ACM SIGCOMM workshop on Hot topics in software defined networking (pp. 19-24). ACM.


\bibitem{c13} Gember, A., Grandl, R., Anand, A., Benson, T., Akella, A. (2012). Stratos: Virtual middleboxes as first-class entities. UW-Madison TR1771.

\bibitem{c14} Greenhalgh, A., Huici, F., Hoerdt, M., Papadimitriou, P., Handley, M., and Mathy, L. (2009). Flow processing and the rise of commodity network hardware. ACM SIGCOMM Computer Communication Review, 39(2), 20-26.


\bibitem{c2} Heorhiadi, V., Reiter, M. K., and Sekar, V. (2012, December). New opportunities for load balancing in network-wide intrusion detection systems. In international conference on Emerging networking experiments and technologies (pp. 361-372). ACM.



\bibitem{c1} Qazi, Z. A., Tu, C. C., Chiang, L., Miao, R., Sekar, V., and Yu, M. (2013, August). SIMPLE-fying middlebox policy enforcement using SDN. In ACM SIGCOMM Computer Communication Review (Vol. 43, No. 4, pp. 27-38). ACM.

\bibitem {c11} Rosen, E., Viswanathan, A., and Callon, R. (2001). Multiprotocol label switching architecture.

\bibitem{c4} Sekar, V., Egi, N., Ratnasamy, S., Reiter, M. K., and Shi, G. (2012, April). Design and implementation of a consolidated middlebox architecture. In USENIX conference on Networked Systems Design and Implementation (pp. 24-24). USENIX Association.

















\end{thebibliography}
\end{document}